%% file: techRep.tex
\documentclass{llncs}
\usepackage[all]{xy}
\xyoption{frame}
\usepackage{algorithmic}
\usepackage{algorithm}
\usepackage{amsfonts}
\usepackage{amssymb}
\usepackage{amsmath}
\usepackage{verbatim}
\usepackage{graphicx}
\usepackage{color}
\usepackage{exscale}

\input{macros}

\title{Interprocedural Dataflow Analysis over Weight Domains with
Infinite Descending Chains\thanks{The
second and fourth authors are supported in part by the DFG project {\em Algorithms for Software Model Checking}.
 The third author is supported in part by Institute for Theoretical Computer Science, project No.~1M0545.}}
\author{
Morten K\"{u}hnrich\inst{2} \and
Stefan Schwoon\inst{1} \and
Ji\v{r}\'{\i} Srba\inst{2} 
\and
Stefan Kiefer\inst{1}}
\institute{
Technische Universit\"{a}t M\"{u}nchen \\ Boltzmannstr. 3, 85748 Garching,
Germany \\ \email{\{kiefer,schwoon\}@in.tum.de}\and
Department of Computer Science, Aalborg University \\
  Selma Lagerl\"{o}fs Vej 300, 9220 Aalborg East, Denmark \\
  \email{\{mokyhn,srba\}@cs.aau.dk} \\
}

\begin{document}

\maketitle
\begin{abstract}
We study generalized fixed-point equations over idempotent semirings
and provide an efficient algorithm for the detection whether a sequence
of Kleene's iterations stabilizes after a finite number of steps.
Previously known approaches considered only bounded semirings
where there are no infinite descending chains. The main novelty
of our work is that we deal with semirings without the
boundedness restriction. Our study is motivated by several applications
from interprocedural dataflow analysis. We demonstrate how the
reachability problem for weighted pushdown automata can be reduced
to solving equations in the framework mentioned above and we
describe a few 
applications to demonstrate its usability.
\end{abstract}

\input{intro}
\input{equations}
\input{wpds}

\input{applications}
\input{conclusion}

\bibliographystyle{plain}
\bibliography{bibliography}

\newpage
\section*{Appendix}
\appendix
\input{proof-theorem}

\input{proof-lemma}
\input{posteq}

\end{document}

%% file: macros.tex
\def\definesto#1{\stackrel{#1}{\hookrightarrow}}
\def\pto#1{\stackrel{#1}{\Rightarrow}}
\def\goesto#1{\stackrel{#1}{\longrightarrow}}

\newcommand{\undefd}{\ensuremath{\bot}}
\newcommand{\naturals}{\ensuremath{\mathbb{N}}}
\newcommand{\integers}{\ensuremath{\mathbb{Z}}}
\newcommand{\rationals}{\ensuremath{\mathbb{Q}}}

\newcommand{\integerswithinfty}{\ensuremath{\mathbb{Z}_\infty}}
\newcommand{\goes}[1]{\ensuremath{\stackrel{#1}{\longrightarrow}}}
\newcommand{\domain}{\ensuremath{D}}
\newcommand{\combine}{\ensuremath{\oplus}}
\newcommand{\combineBig}{\ensuremath{\bigoplus}}
\newcommand{\extend}{\ensuremath{\otimes}}

\newcommand{\combineNeutral}{\ensuremath{\overline{0}}}
\newcommand{\extendNeutral}{\ensuremath{\overline{1}}}
\newcommand{\vX}{\vec X}
\newcommand{\vf}{\vec f}

\newcommand{\vzero}{\vec \combineNeutral}
\newcommand{\gfp}{{\mathit{gfp}}}
\newcommand{\vv}{\vec v}
\newcommand{\vectors}{V}
\newcommand{\yi}{\mathsf{Y}}
\newcommand{\vkappa}{\boldsymbol{\kappa}}
\newcommand{\vlambda}{\boldsymbol{\lambda}}
\newcommand{\vmu}{\boldsymbol{\mu}}
\newcommand{\ks}[1]{\vkappa^{(#1)}}

\def\A{{\mathcal{A}}}

\def\places{P}
\def\stackalph{\Gamma}
\def\transitionrules{\Delta}

\def\wpds{\mathcal{W}}
\def\ring{\mathcal{S}}
\def\weight{d}

\def\set#1{\{#1\}}

\newcommand{\todo}[1]
{\marginpar{\baselineskip0ex\rule{2,5cm}{0.5pt}\\[0ex]{\tiny\textsf{#1}}}}

\newcommand{\node}[1]{*+++[o][F-]{\makebox(3,3){$#1$}}}
\newcommand{\nodel}[1]{*+++[o][F-]{\makebox(12,12){$#1$}}}
\newcommand{\nodedouble}[1]{*+++[o][F=]{\makebox(3,3){$#1$}}}
\newcommand{\nodedoublel}[1]{*+++[o][F=]{\makebox(12,12){$#1$}}}

\def\sema#1{\lbrack{#1}\rbrack}
\def\sint{\ring_{\mathit int}}

%% file: intro.tex
\section{Introduction}

Weighted pushdown systems~\cite{RSJM05} are a suitable model for analyzing
programs with procedures. They have been used successfully in a number of
applications, e.g.\
BDD-based model checking~\cite{SBSE07,EKS06b},
trust-management systems~\cite{JSWR06},
path optimization~\cite{LLPL06},
and interprocedural data\-flow analysis (see~\cite{RLK07} for a survey).

The main idea is that the
transitions of a pushdown system are labelled with values from a given
data domain (e.g. natural numbers). These values can be composed when
executed in sequence (e.g. using the addition on natural numbers)
and one is then interested in a number of verification questions like
reachability of a given configuration with the combined value
over all paths leading into this configuration (e.g. by taking the
minimum value over all such paths). It has been shown that there
are efficient polynomial time
algorithms for answering these questions~\cite{RSJM05}.

In this paper, we contribute to the research in this area. 
We first draw a connection between reachability in weighted
pushdown systems (WPDS) over an idempotent semiring and solving fixed-point
equations over the same semiring.
Unlike related work, we allow for infinite descending
chains in our semirings (our approach e.g. includes the integer semiring).
Due to this reason, the system of equations constructed from a WPDS may
not have a solution. We therefore provide an efficient algorithm
that either determines the solution or detects the presence of an infinite
descending chain. In the latter case, we output
some component (variable) of the system affected by the problem.
So on one hand we treat domains with infinite descending chains
but on the other hand, two restrictions are necessary to make
this possible. However, as argued in Section~\ref{sec:apps},
the framework still includes a number of interesting applications.

For better readability some proofs have been moved to an appendix.


\subsection{Dataflow Analysis and Fixed-Point Equations}
\label{sub:intro-dataflow}

Static analysis gathers information about a program without executing
it. Data\-flow analysis is an instance of static analysis: it reasons
about run-time values of variables or expressions. 
More to the point, we desire to establish facts that hold at some control
point whenever an execution reaches it.

Most approaches to dataflow analysis reduce the problem
(explicitly or implicitly) to solving
a system of fixed-point equations over some algebraic structure, e.g.\ a
lattice or a semiring.
They map the control-flow graph of a program to an equation system
\mbox{$\vX = \vf(\vX)$}, where
the vector $\vX = (X_1, \ldots, X_n)$ stands for the nodes in the control flow
graph, and takes values from some dataflow domain.
The vector $\vf = (f_1, \ldots, f_n)$ stands for the edges in the
graph, i.e., the {\em transfer function}~$f_i(\vX)$
describes the effect of the program on~$X_i$ in terms of the other dataflow
values. Under certain conditions (e.g., the functions~$f_i$ are
distributive) the desired dataflow information is precisely the
greatest solution of the system \mbox{$\vX = \vf(\vX)$},
i.e., the greatest fixed point $\gfp(\vf)$ of~$\vf$ \cite{NNH99,SP81}.

There is a large body of literature dealing with dataflow analysis along
these lines. Of particular interest to us are interprocedural
analyses. The seminal work of Sharir and Pnueli~\cite{SP81} shows how to
set up an equation system that captures only the \emph{interprocedurally valid}
paths, i.e.\ those paths in which all return statements lead back to the
site of the most recent call. However, \cite{SP81} computes only one dataflow
value for each program point, merging together all the paths that reach it,
regardless of the calling context. In \cite{RSJM05} a generalization was provided,
where the solution of the equations computes a solution for each \emph{configuration}, where configuration denotes a program point together with its calling
context. Thus, \cite{RSJM05} allows to distinguish dataflow values for
different, arbitrary calling contexts. (The merged information can still be
obtained as a special case.)  The results of \cite{RSJM05} were phrased in
terms of weighted pushdown systems (WPDS), and we will adopt this notion in
our paper.

If the dataflow domain satisfies the so-called
{\em descending chain condition}
(i.e. each infinite descending chain eventually becomes stationary),
$\gfp(\vf)$ can be obtained by {\em Kleene's iteration}:
Let $\combineNeutral$ be the greatest domain element, and
$\vzero = (\combineNeutral, \ldots, \combineNeutral)$.
Then Kleene's fixed-point theorem guarantees that
the sequence $\vzero, \vf(\vzero), \vf(\vf(\vzero)), \ldots$ reaches
$\gfp(\vf)$ after finitely many steps. Both \cite{SP81} and \cite{RSJM05}
require the descending chain condition.

However, the descending chain condition does not always hold.
For example, the lattice of non-positive integers with $\sqcap = \min$
and $\sqcup = \max$
does not satisfy the condition because of the infinite
descending chain $0, -1, -2, \ldots$.
In fact, this chain arises when doing Kleene's iteration on the equation
$X = f(X)$ where $f(X) = \min(X,X-1)$. More to the point, Kleene's iteration
on $f$ would fail to terminate.
We will show how to overcome this problem.

Previous work (e.g., \cite{RSJM05})
has shown that many important analysis problems
can be phrased as equation systems, where $\vf(\vX)$ contains polynomials
over {\em idempotent semirings}.
By polynomial, we mean an expression that is built up from variables, constant elements,
 and the semiring operations `$\combine$' (combine) and `$\extend$' (extend).

Recently, fixed-point equations over idempotent semirings
have been studied intensively. While the classical solution is to use
Kleene's iteration or chaotic iteration, recent work has
proposed faster algorithms and better convergence results based on
Newton's method~\cite{HK99,EKL07:stacs,EKL07:dlt,EKL08:icalp}.
In these works, the boundedness condition is dropped, but replaced
by another condition called $\omega$-continuity, requiring that the
infimum of every infinite set exists, thus ensuring that a greatest
fixed point can always be found. Our work does not require this condition,
and a greatest fixed point is not always guaranteed to exist (but our
algorithm detects such a case and reports it). The penalty for this
is that a different kind of restriction has to be introduced:
we require that semirings are
totally ordered and that ``extend preserves
inequality'', i.e., $a\extend c\ne b\extend c$ for $a\ne b$ and
$a,b,c\ne\combineNeutral$.

Our algorithm executes Kleene's
iteration, and if the iteration terminates, it outputs the greatest
fixed point. If Kleene's iteration fails to terminate, our algorithm will
detect this and still terminate, indicating a responsible variable (a so-called
\emph{witness component}).

The work closest to ours is the one by Gawlitza and
Seidl~\cite{GawlitzaSeidl:Precise-Fixpoint}, who consider systems
of equations over the integer semiring. Our algorithm can be seen
as a generalization of one of their algorithms 
to totally ordered semirings where extend preserves inequality
and to equations over arbitrary polynomials. Moreover, we provide
a direct and self-contained proof of the result.
Another related work is by Leroux and Sutre~\cite{LS:SAS:07}.
They present 
an algorithm for computing least fixed-points
for monotone \emph{bounded-increasing} functions over integers.
On one hand they consider more general functions
like e.g. factorials, on the other hand the minimum and maximum
functions are not bounded-increasing according to their definition.
As a result, their algorithm
is not applicable in our setting of weighted pushdown systems.

We proceed as follows:
In Section~\ref{sec:eq}, we provide a new algorithm
for solving fixed-point equations.
Using this result, we design
a new algorithm for interprocedural dataflow analysis in Section~\ref{sec:wpds},
which is based on WPDS~\cite{RSJM05} and still requires a polynomial
number of semiring operations. Like previous work on WPDS,
the algorithm allows to compute dataflow information for each
configuration (if desired). Due to the properties of the systems we
handle, our algorithm either returns a solution (if it exists) or reports
that none exists (usually indicating an error in the program).
We provide several applications of our theory in Section~\ref{sec:apps}.

%% file: equations.tex
\newcommand{\vars}{\ensuremath{\mathcal{X}}}
\newcommand{\vg}{\vec g}
\newcommand{\vu}{\vec u}
\section{Fixed-Point Equations over Idempotent Semirings}
\label{sec:eq}

In this section we shall study fixed-point equations over idempotent
semirings and Kleene's iterations over vectors of polynomials.

%
\begin{definition}[Idempotent Semiring] \label{def:semiring}
An \emph{idempotent semiring} is a 5-tuple
$\ring=(\domain, \combine, \extend, \combineNeutral, \extendNeutral)$
where $\domain$ is a set called the \emph{domain},
$\combineNeutral, \extendNeutral \in \domain$, and
the binary operators \emph{combine} `$\combine$'
and  \emph{extend} `$\extend$' on $\domain$ satisfy:
\begin{enumerate}
\item $(\domain, \combine)$ is a commutative monoid with $\combineNeutral$
as its neutral element and $(\domain, \extend)$ is a monoid
with $\extendNeutral$ as its neutral element,
\item extend distributes over combine, i.e.,
$\forall a, b, c \in \domain:
a \extend (b \combine c) = (a \extend b) \combine (a \extend c)$ and
$(a \combine b) \extend c = (a \extend c) \combine (b \extend c)$,
\item $\combineNeutral$ is an annihilator for extend, i.e.,
$\forall a \in \domain: a \extend \combineNeutral = \combineNeutral \extend a = \combineNeutral$, and
\item every $a \in \domain$ is idempotent w.r.t. combine, i.e., $\forall a \in \domain: a \combine a = a$.
\end{enumerate}
\end{definition}

\begin{definition}[Ordering]\label{def:ordering}
We write $a \sqsubseteq b$ for $a,b \in \domain$ whenever $a \combine b = a$.
\end{definition}

As we are mainly interested in algorithmic verification approaches,
we shall implicitly consider only \emph{computable} semirings
where the elements from the domain are effectively representable,
operations combine and extend are algorithmically computable and
the test on equality is decidable. We will
use the big-$O$-notation for complexity upper-bounds, though
it should be always interpreted relative to the complexity
of the semiring operations. In the semirings considered
in our applications,
we can assume that all operations can be performed in $O(1)$ time.
Hence the big-$O$-notation for the semirings mentioned in this
paper corresponds to the standard asymptotic complexity.

\begin{lemma} \label{lem:properties}
(i) For all $a,b \in \domain$ it holds that $a \combine b \sqsubseteq a$.
(ii) For all $a,b,c \in \domain$ it holds that if
$a \sqsubseteq b$ then $a \extend c \sqsubseteq b \extend c$.

\end{lemma}

The proof of Lemma~\ref{lem:properties} is straightforward.
We shall now define an additional condition on the extend operator
that will be used later on in this section.

\begin{definition}[Extend Preserves Inequality]
Given an idempotent semiring we say that \emph{extend preserves
inequality} if $a \not= b$ implies that $a \extend c \not= b \extend c$
for any $a, b, c \in \domain \smallsetminus \{\combineNeutral\}$.
\end{definition}

\begin{example} \label{ex:semirings}
The tuple $\sint
= (\integerswithinfty,\min,+,\infty,0)$ is an idempotent
semiring.
The domain are the integers extended
with infinity $\integerswithinfty = \integers \cup \set{\infty}$
where
$\min(\infty, a) = \min(a, \infty) = a$ and
$a + \infty = \infty + a = \infty$ for all $a \in \integerswithinfty$.
Combine is the minimum and extend is the usual addition on integers.
It is easy to see that $\sint$ meets
the requirements of Definition \ref{def:semiring}. It moreover preserves
inequality because the addition does so, and $\sqsubseteq$ is a total order.

Another example of an idempotent semirings is
$\ring_{\mathit rat}=(\rationals[0,1],\max,*,0,1)$
which is the semiring defined over the rationals in the interval from
$0$ to $1$. Here combine is the maximum and extend is the multiplication
on rationals.
This semiring $\ring_{\mathit rat}$ also meets the requirements of
Definition \ref{def:semiring}, extend preserves inequality
 and $\sqsubseteq$ is a total order.
\qed
\end{example}


In what follows
we fix an idempotent semiring
$\ring=(\domain, \combine, \extend, \combineNeutral, \extendNeutral)$.
We often omit the $\extend$ sign in ``products'', i.e., we write $ab$ for $a \extend b$.
We also fix a set $\vars = \{X_1, \ldots, X_n\}$ of variables.
Now we define vectors of polynomials over~$\ring$ and their fixed points following~\cite{EKL07:dlt}.

Let $\vectors = \domain^n$ denote the set of {\em vectors} over $\ring$.
We use bold letters to denote vectors, e.g., $\vv = (\vv_1, \ldots, \vv_n)$.
We also write $\vX = (X_1, \ldots, X_n)$ to arrange the variables from~$\vars$ in a vector.
We extend $\sqsubseteq$ to vectors by setting $\vu \sqsubseteq \vv$ if $\vu_i \sqsubseteq \vv_i$ for all $1 \le i \le n$.

A {\em monomial} is a finite expression
$a_1 X_{i_1} a_2 X_{i_2} \cdots a_s X_{i_s} a_{s+1}$
where $s \geq 0$, $a_1,\ldots,a_{s+1}\in \domain$ and $X_{i_1},\ldots,X_{i_s} \in \vars$.
A {\em polynomial} is an expression of the form $m_1 \combine \cdots \combine m_s$
where $s \geq 0$ and $m_1, \ldots, m_s$ are monomials.
The value of a monomial $m = a_1 X_{i_1} a_2 \cdots a_s X_{i_s} a_{s+1}$ at $\vv$
 is \mbox{$m(\vv)= {a_1 \vv_{i_1} a_2 \cdots a_s \vv_{i_s} a_{s+1} \in \domain}$}.
The value of a polynomial $f = m_1 \combine \cdots \combine m_s$
 at~$\vv$ is $f(\vv) = m_1(\vv) \combine \cdots \combine m_s(\vv)$.
A polynomial induces a mapping from $\vectors$ to $\domain$ that assigns to $\vv$ the element $f(\vv)$.
A vector of polynomials $\vf=(\vf_1, \ldots, \vf_n)$
is an $n$-tuple of polynomials;
 it induces a mapping from $\vectors$ to $\vectors$
 that assigns to a vector $\vv$ the vector $\vf(\vv) = (\vf_1(\vv), \ldots, \vf_n(\vv))$.
A {\em fixed point of~$\vf$} is a vector $\vv$ that satisfies $\vv = \vf(\vv)$.
A greatest fixed point of~$\vf$ is a fixed point $\vv$ such that $\vv' \sqsubseteq \vv$ holds for all other fixed points~$\vv'$.
The size $K(\vf)$ of a vector of polynomials $\vf$ is the total number of $\combine$ and $\extend$ operators in~$\vf$.
In particular, given a vector $\vv$, it takes $O(K(\vf))$ time to compute~$\vf(\vv)$.

\begin{example} \label{ex:run-vector}
 Consider the semiring $\sint$ from Example~\ref{ex:semirings}.
 Let $\vars = \{X_1, X_2, X_3\}$.
 Then $\vf = ( -2 \combine X_2 \extend X_3, \ X_3 \extend 1, \ X_1 \combine X_2 )$
  is a vector of polynomials over~$\sint$.
 It can be rewritten as
  $\vf = ( \min\{-2, X_2 + X_3\}, \ X_3 + 1, \ \min\{X_1,X_2\} )$.
 The size $K(\vf)$ equals $4$.
\qed
\end{example}

It is easy to see that polynomials are monotone and continuous mappings w.r.t.~$\sqsubseteq$, see Lemma~\ref{lem:properties}.
Kleene's theorem can then be applied (see e.g.~\cite{Kui}), which leads to
the following proposition.


\begin{proposition}\label{prop:kleene}
 Let $\vf$ be a vector of polynomials.
 Let the {\em Kleene sequence} $(\ks{k})_{k \in \naturals}$ be defined by $\ks{0} = \vzero$ and $\ks{k+1} = \vf(\ks{k})$.
 \begin{itemize}
  \item[(a)]
   We have $\ks{k+1} \sqsubseteq \ks{k}$ for all $k\in\naturals$.
  \item[(b)]
   If a greatest fixed point exists then it is the infimum 
   of~\mbox{$\{\ks{k} \mid k \in \naturals\}$}.
  \item[(c)]
   If the infimum of~$\{\ks{k} \mid k \in \naturals\}$ exists then it is the greatest fixed point.
 \end{itemize}
\end{proposition}

Proposition~\ref{prop:kleene} is the mathematical basis for
the classical fixed-point iteration:
apply~$\vf$ until a fixed point is reached, which is, by Proposition~\ref{prop:kleene}~(c),
 the greatest fixed point of~$\vf$.
We call this method {\em Kleene's iteration}.
In general, Kleene's iteration does not always reach a fixed point.
Some equations, like
 \mbox{$X = X \extend (-1)$} over $\sint$,
do not have any (greatest)
fixed point, other equations might have a greatest
fixed point but it is not achievable in a finite number of
Kleene's iterations (consider for example the above equation
but over the semiring $\sint$ extended with the
element $-\infty$).
It is not a priori clear how to detect whether Kleene's iteration
terminates, i.e., computes the greatest fixed point in a finite
number of iterations.


Algorithm~\ref{alg:computation} (called ``safe Kleene's iteration'') solves this problem.
If Kleene's iteration reaches the greatest fixed point, then the algorithm computes it.
Otherwise it outputs a {\em witness component} where
Kleene's iteration does not terminate.
Formally, a witness component is defined as follows.
\begin{definition}[Witness Component]\label{def:witness}
 Let $\vf$ be a vector of polynomials over an idempotent semiring.
 A component $i$ \ ($1 \le i \le n$) is a {\em witness component} if $\{\ks{k}_i \mid k \ge 0\}$ is an infinite set.
\end{definition}
In our applications, the presence of a witness component pinpoints a problem of the analyzed model which the user may want to fix.
More details are given in Section~\ref{sec:apps}.

Algorithm~\ref{alg:computation} is based on the generalized Bellman-Ford algorithm of~\cite{GawlitzaSeidl:Precise-Fixpoint}
for $\sint$ and generalizes it further
to totally ordered semirings where extend preserves inequality
and to equations over arbitrary polynomials.

\begin{algorithm}[!ht]
\caption{Safe Kleene's iteration}\label{alg:computation}
\begin{algorithmic}[1]
\REQUIRE A vector of polynomials $\vf = (\vf_1, \ldots, \vf_n)$ over an idempotent semiring
$\ring=(\domain, \combine, \extend, \combineNeutral, \extendNeutral)$
s.t. $\sqsubseteq$ is a total order and
where extend preserves inequality.
\ENSURE Greatest fixed point of~$\vf$ or a witness component.  \\
\STATE {$\ks{0} := \vzero$}
\FOR {$k := 1$ {\bf to} $n+1$} \label{line:main-loop}
 \STATE $\ks{k} := \vf(\ks{k-1})$
\ENDFOR
\IF {$\exists i$ with $1 \le i \le n$ such that $\ks{n+1}_i \ne \ks{n}_i$} \label{line:check-if-equal}
  \STATE {{\bf return} ``Kleene's iteration does not terminate. Component $i$ is a witness.''} \label{line:witness}
\ELSE
  \STATE {{\bf return} ``The vector $\ks{n}$ is the greatest fixed point.''}
\ENDIF
\end{algorithmic}
\end{algorithm}

\begin{theorem} \label{thm:dg}
Algorithm~\ref{alg:computation} is correct and terminates in time $O(n \cdot K(\vf))$.
\end{theorem}

Algorithm~\ref{alg:computation} on its own is very straightforward,
and its proof for polynomials of degree only 1 would directly mimic
the proof of Bellman-Ford algorithm. Our contribution is that we prove that
it works also for polynomials of higher degrees where more involved
technical treatment is necessary.
Full details can be found in Appendix~\ref{proof:thm}.

\begin{remark} \label{rem:all-witnesses}
 In the integer semiring $\sint$,
 Algorithm~\ref{alg:computation} can be extended such that
it computes {\em all} witness components and for the remaining terminating
components returns the exact value.
This is done as follows.
The main loop on lines 2--4 is run once again, but the components that still change are assigned
 a new semiring element ``$-\infty$'' on which the operators ``$+$'' and ``$\min$'' act as expected.
Thus, $-\infty$ may be propagated through the components during the repetition of the main loop.
At the end, all components that are not $-\infty$ have reached their final value,
 all others can be reported as witness components.
For details see~\cite{GawlitzaSeidl:Precise-Fixpoint}.
\end{remark}

\begin{example} \label{ex:run-kleene}
 Consider again the vector of polynomials from Example~\ref{ex:run-vector}:
 \[
  \vf = ( \min\{-2, X_2 + X_3\}, \ X_3 + 1, \ \min\{X_1,X_2\} ) \ .
 \]
 Kleene's iteration produces the following Kleene sequence:
$\ks{0} = (\infty, \infty, \infty)$,
$\ks{1} = (-2, \infty, \infty)$,
$\ks{2} = (-2, \infty, -2)$,
$\ks{3} = (-2, -1, -2)$,
$\ks{4} = (-3, -1, -2)$.
As $\ks{3}_1 = -2 \ne -3 = \ks{4}_1$, Alg.~\ref{alg:computation} returns the first component as a witness. 
\qed
\end{example}

Notice that Algorithm~\ref{alg:computation} merely indicates whether
a greatest fixed point can be found \emph{using Kleene's iteration} or not.
Even if Algorithm~\ref{alg:computation} outputs a witness component,
a greatest fixed point may still exist (and be found by other means).
An example is a semiring over the reals which can admit
the sequence $1/2^n$ for some variable. This sequence converges to~$0$,
but Kleene's iteration fails to detect this. Nevertheless, for some
semirings like $\sint$ used in our applications,
we can make the following stronger statement.
\begin{corollary} \label{cor:intring}
Algorithm~\ref{alg:computation} applied to polynomials over the semiring
$\sint$ finds the greatest fixed point iff it exists.
If it does not exist, all witness components can be explicitly marked.
\end{corollary}
\begin{proof}
In $\sint$ a component is a
witness component iff Kleene's iteration does not terminate in that
component. The rest follows
from Definition~\ref{def:witness}, Proposition~\ref{prop:kleene}
and Remark~\ref{rem:all-witnesses}. \qed
\end{proof}




%% file: wpds.tex
\section{Weighted Pushdown Systems}
\label{sec:wpds}

In this section we will use the fixed-point equations studied
in the previous section for reasoning about
properties of weighted pushdown systems (WPDS)~\cite{RSJM05}.
We are interested in applying Theorem~\ref{thm:dg} to weighted
pushdown systems; therefore we implicitly consider only semirings
that are totally ordered, and where extend preserves inequality.

\begin{definition}[Weighted Pushdown System]
A weighted pushdown system
is a 4-tuple $\wpds=(P, \Gamma, \Delta, \ring)$, where $P$ is a finite
set of \emph{control states}, $\Gamma$ is a finite \emph{stack alphabet},
$\Delta \subseteq (P \times \Gamma) \times \domain \times
(P \times \Gamma^{*})$ is a finite set of \emph{rules}, and
$\ring=(\domain, \combine, \extend, \combineNeutral, \extendNeutral)$
is an idempotent semiring.
\end{definition}

We write $pX \definesto{\weight} q \alpha$
whenever $r=(p, X, \weight, q, \alpha) \in \Delta$ and call
$d$ the \emph{weight} of~$r$, denoted by~$d_r$.
We consider only rules
where $|\alpha| \leq 2$.
(It is well-known that every WPDS can be translated into a
one that obeys this restriction and is larger by only a
constant factor, see, e.g., \cite{Sch02b}. The reduction
preserves reachability.)
We let the symbols $X, Y, Z$ range over $\stackalph$
and $\alpha,\beta,\gamma$ range over $\stackalph^*$.

\begin{example} \label{ex:wpds}
As a running example in this section,
we consider a weighted pushdown system over the semiring
with both positive and negative integers as weights, i.e.\
$\wpds_{ex}=(\{p,q\},\{X,Y\},\Delta_{ex},\sint)$, where
$\Delta_{ex} = \{pX\definesto{1}qY,\ \ pX\definesto{1}pXY,\ \ pY\definesto{1}p,\ \ qY\definesto{-2}q\}$.
\qed
\end{example}

A \emph{configuration} of a weighted pushdown system
$\wpds$ is a pair $p\gamma$ where
$p \in P$ and $\gamma \in \Gamma^*$. A transition relation
$\pto{}$ on configurations is defined by
$pX\gamma \pto{r} q\alpha\gamma$ iff $\gamma\in\Gamma^*$ and
there exists $r\in\Delta$, where
$r = (pX \definesto{d} q\alpha)$.
We annotate $\pto{}$ with the rule $r \in \Delta$
which was used to derive the conclusion.
If there exists a sequence of configurations $c_0,\ldots,c_n$
and rules $r_1,\ldots,r_n$ such that $c_{i-1}\pto{r_i}c_i$ for
all $i=1,\ldots,n$, then we write $c_0\pto{\sigma}c_n$, where
$\sigma:=r_1\ldots r_n$. The weight of $\sigma$ is defined as
$v(\sigma)=d_{r_1}\extend\cdots\extend d_{r_n}$.
By definition $v(\epsilon) = \extendNeutral$.

Let $c,c'$ be two configurations and $\sigma\in\Delta^*$ such that
$c\pto{\sigma}c'$. We call $c$ a \emph{predecessor} of $c'$ and
$c'$ a \emph{successor} of $c$. In the following, we will consider the
problem of computing the set of all predecessors $\mathit{pre}^*(c_f)$
and successors $\mathit{post}^*(c_f)$
for a given configuration $c_f$. Due to space limitations we provide
the full treatment only for the predecessors; the computation of
successors is analogous and it is provided in Appendix~\ref{app:post}.

Let us fix a WPDS~$\wpds$ and a \emph{target configuration}
$c_f$, where $c_f=p_f\epsilon$ for some control state~$p_f$.
For any configuration $c$ of $\wpds$,
we want to know the minimal weight of a path from~$c$ to $c_f$. If
a path of minimal weight does not exist for every~$c$, we want to
detect such a case.
In our applications (see Section~\ref{sec:apps}),
this situation usually indicates the existence of an error.
\begin{remark}
In the literature, it is more common to consider
a \emph{regular} set $C$ of target configurations.
This problem, however,
reduces to the one with only a single target configuration~$c_f$.
The reduction can be achieved by extending~$\wpds$ with additional `pop'
rules that simulate a finite automaton for $C$; the `pop' rules will
succeed in reducing the stack to~$c_f$ iff they begin with a configuration
in~$C$. For details, see~\cite{RSJM05}, Section 3.1.1.
\end{remark}

At an abstract level, we are interested in solutions for the following
equation system, in which each configuration~$c$ is represented
by a variable~$\sema{c}$. Intuitively, the greatest solution (if it exists) for
the variable $\sema{c}$ will correspond to the minimum (w.r.t. the combine operator)
of accumulated weights over
all paths leading from the configuration $c$ to $c_f$.

\begin{equation}\sema{c}=I(c)\combine\bigoplus_{c\pto{r}c'}(d_r\extend \sema{c'}),
\qquad \hbox{where\ } I(c):=\begin{cases}\extendNeutral& \hbox{if $c=c_f$} \\
\combineNeutral & \hbox{otherwise }\end{cases}\label{e:fofc}\end{equation}
Let us consider the Kleene sequence $(\ks{k})_{k \in \naturals}$ for
(\ref{e:fofc}). By $\ks{k}_{\sema{c}}$ we denote the entry for configuration~$c$
in the $k$-th iteration of the Kleene sequence.

\begin{lemma}
\label{lem:fofc}
For $k\ge1$ and any configuration~$c$, the following holds
$$\ks{k}_{\sema{c}}=\bigoplus\{\,v(\sigma)\mid c\pto{\sigma}c_f, \ |\sigma|<k\,\}\ .$$
\end{lemma}

Thus, $\sema{c}$ is a witness component of~(\ref{e:fofc}) iff no path of minimal
weight exists, because it is possible to construct longer and longer paths
with smaller and smaller weights.
On the other hand, if (\ref{e:fofc}) has a greatest fixed point, then
the fixed point at~$\sema{c}$ gives the combine of the weights of all
sequences leading from $c$ to~$c_f$, commonly known as the
\emph{meet-over-all-paths}.
However, (\ref{e:fofc}) defines an infinite system of equations,
which we cannot handle directly. In the following, we shall derive a
\emph{finite} system of equations, from which we can determine the
greatest fixed point of (\ref{e:fofc}) or the existence of a witness component.

\begin{definition}[Pop Sequence]
Let $p,q$ be control states and $X$ be a stack symbol. A \emph{pop sequence}
for $p,X,q$ is any sequence $\sigma\in\Delta^*$ such that
$pX\pto{\sigma}q\epsilon$.
\end{definition}

Let us consider the following polynomial equation system, in which the variables
are triples $\sema{pXq}$, where $p,q$ are control states and $X$ a
stack symbol:
\begin{equation}
\sema{pXq} =\hspace{-5mm} \bigoplus_{(pX\definesto{d}q\epsilon) \in \Delta}
\hspace{-5mm} d \ \ \
\combine \hspace{-2mm} \bigoplus_{(pX\definesto{d}rY) \in \Delta}
\hspace{-4mm}
\big(d\extend \sema{rYq}\big) \
\combine
\hspace{-5mm}
\bigoplus_{(pX\definesto{d}rYZ) \in \Delta}
\hspace{-5mm}
\bigg(d\extend \bigoplus_{s\in P} \big(\sema{rYs}\extend
\sema{sZq}\big)\bigg) \ .
\label{e:goft}
\end{equation}
Intuitively, Equation~(\ref{e:goft}) lists all the possible ways in which
a pop sequence for $p,X,q$ can be generated and computes the values accumulated
along each of them.

\begin{example} \label{ex:wpdseqns}
Let us consider the WPDS $\wpds_{ex}$ from Example~\ref{ex:wpds}.
Here, the scheme presented in~(\ref{e:goft}) yields a system with
eight variables and equations, four of which are reproduced below.
$$\begin{array}{rcl@{\qquad}rcl}
\sema{pXp} & = & \min\{1+\sema{qYp},\ 1+\sema{pXp}+\sema{pYp},
            \ 1+\sema{pXq}+\sema{qYp}\} &
\sema{pYp} & = & 1 \\
\sema{pXq} & = & \min\{1+\sema{qYq},\ 1+\sema{pXp}+\sema{pYq},
            \ 1+\sema{pXq}+\sema{qYq}\} &
\sema{qYq} & = & -2
\end{array}$$
Notice that the other four variables would be simply assigned to the
$\combineNeutral$ element, in this case $\infty$.
\qed
\end{example}

We now examine the Kleene sequence $(\ks{k})_{k \in \naturals}$
for~(\ref{e:goft}).
\begin{lemma}
\label{lem:goft}
For any $k\ge1$, control states $p,q$, and stack symbol~$X$,
$$\bigoplus\{\,v(\sigma)\mid c\pto{\sigma}c_f, \ |\sigma|\le 2^{k-1}\,\}
\sqsubseteq
\ks{k}_{\sema{pXq}}
\sqsubseteq
\bigoplus\{\,v(\sigma)\mid c\pto{\sigma}c_f, \ |\sigma|\leq k-1\,\}\ .$$
\end{lemma}

Thus, $\sema{pXq}$ is a witness component of~(\ref{e:goft}) iff
no minimal-weight pop sequence exists for $p,X,q$. On the other hand,
if no witness component exists, then the value of $\sema{pXq}$ in the greatest
fixed point denotes the combine of the weights of all pop sequences for
$p,X,q$.

We now show how (\ref{e:goft}) can be used to derive statements
about~(\ref{e:fofc}).
Let a configuration $c=pX_1\ldots X_n$ be a predecessor
of $c_f$. Then any sequence $\sigma$ leading from $c$ to $c_f$ can be
subdivided into subsequences $\sigma_1,\ldots,\sigma_n$
and there exist states
$p=:p_0,p_1,\ldots,p_{n-1},p_n:=p_f$ such that $\sigma_i$ is
a pop sequence for $p_{i-1},X_i,p_i$, for all $i=1,\ldots,n$.
As a consequence, we can obtain a solution for~(\ref{e:fofc})
from a solution for~(\ref{e:goft}): suppose that $\vlambda$ is
the greatest fixed point of~(\ref{e:goft}), and let $\vmu$
be a vector of configurations as follows:
\begin{equation}
\vmu_{\sema{c}} = \bigoplus_{p_1,\ldots,p_{n-1}}
\big( \vlambda_{\sema{pX_1p_1}}\extend\cdots\extend\vlambda_{\sema{p_{n-1}X_np_f}} \big), \qquad \hbox{for $c=pX_1\ldots X_n$} \ .
\label{e:vmu}
\end{equation}
It is easy to see that (\ref{e:vmu}) ``sums up'' all possible paths
from $c$ to $c_f$, and therefore yields the meet-over-all-paths for~$c$.
Thus, $\vmu$ is a solution (greatest fixed point)
of~(\ref{e:fofc}). On the other hand,
if (\ref{e:fofc}) has a witness component, then (\ref{e:goft}) must also
have one.

\begin{theorem} \label{thm:fmh}
Applying Algorithm~\ref{alg:computation} to (\ref{e:goft}) either
yields a witness component or, via (\ref{e:vmu}), the greatest fixed point
of (\ref{e:fofc}).
\end{theorem}

\begin{example} \label{ex:wpdssolv}
Once again, consider $\wpds_{ex}$ from Example~\ref{ex:wpds}
and the equation system from Example~\ref{ex:wpdseqns}.
Here, the Kleene sequence quickly converges to the values
$1$ for $\sema{pYp}$, \ $-2$ for $\sema{qYq}$, and $\infty$
for all other variables except $\sema{pXq}$, which turns out
to be a witness component of (\ref{e:goft}). Indeed, one can construct a series
of pop sequences for $p,X,q$ with smaller and smaller weights, e.g.
$pX\pto{1}qY\pto{-2}q\epsilon,$ and
$pX\pto{1}pXY\pto{1}qYY\pto{-2}qY\pto{-2}q\epsilon,$  and etc. 
with weights $-1$, $-2$ etc. If $c_f=q\epsilon$, this implies
that, e.g., $pX$ is a witness component
of~(\ref{e:fofc}). On the other hand, $qY$ or $qYY$ would not be a witness
components, because their values in (\ref{e:vmu}),
would not be affected by the variable $\sema{pXq}$ and evaluate
to $-2$ and $-4$, respectively.
\qed
\end{example}

\begin{remark}\label{rem:complexity}
The size of the equation system (\ref{e:goft}) is polynomial in $\wpds$.
Notice that it makes sense to generate equations only for such triples
$p,X,q$ in which $pX$ occurs on the left-hand side or right-hand side
of some rule. Under this assumption, the number of equations in
(\ref{e:goft}) is $\mathcal{O}(|P|\cdot|\Delta|)$, and its overall size
is $\mathcal{O}(|P|^2\cdot|\Delta|)$, the same complexity as in
the algorithms for computing predecessors in~\cite{EsparzaHRS00}.
According to Theorem~\ref{thm:dg}, Algorithm~\ref{alg:computation}
therefore runs in
$\mathcal{O}(|P|^3\cdot|\Delta|^2)$ time on~(\ref{e:goft}).
For any configuration~$c$ of interest, the value $\vmu_c$ in~(\ref{e:vmu})
can be easily
obtained 
from the result of Algorithm~\ref{alg:computation}. See also the $\wpds$-automaton technique in the subsection to follow.
A similar conclusion about the complexity of the algorithm for
computing successors can be drawn thanks to the (linear) connection
between forward and backward reachability analysis described in
Appendix~\ref{app:post}.
\end{remark}

\subsection{Weighted Automata}
\label{sub:wa}

For (unweighted) pushdown systems, it is well-known that reachability
preserves regularity; in other words, given a regular set of configurations,
the set of all predecessors resp.\ successors is regular. Moreover, given
a finite automaton recognizing a set of configurations, automata recognizing
the predecessors or successors can be constructed in polynomial time
(see, e.g., \cite{EsparzaHRS00}).

It is also known that the results carry over to weighted pushdown systems
provided that the semiring is \emph{bounded}, i.e., there are no infinite
descending chains w.r.t. $\sqsubseteq$~\cite{RSJM05}.
For this purpose, so-called
\emph{weighted automata} are employed.

\begin{definition}[Weighted $\wpds$-Automaton]
Let $\wpds=(P,\Gamma,\Delta,\ring)$ be a pushdown system over
a bounded semiring $\ring$.
A \emph{$\wpds$-automaton} is a 5-tuple
$\A=(Q, \stackalph, \to, \places, F)$
where $Q$ is a finite set of \emph{states},
$\mathord{\to} \subseteq Q \times \Gamma \times\domain \times Q$
is a finite set of \emph{transitions}, 
$\places \subseteq Q$, i.e.\ the control states of $\wpds$,
are the set of \emph{initial states}
and $F \subseteq Q$ is a set of \emph{final (accepting) states}.

Let $\pi=t_1\ldots t_n$ be a path in $\A$,
where $t_i=(q_i,X_i,d_i,q_{i+1})$ for all $1\le i\le n$.
The weight of~$\pi$ is defined as $v(\pi):=d_1\extend\cdots\extend d_n$.
If $q_1\in P$ and $q_{n+1}\in F$, then we say that $\pi$ \emph{accepts}
the configuration $q_1X_1\ldots X_n$. Moreover, if $c$ is a configuration,
we define $v_\A(c)$ as the combine of all $v(\pi)$ such that
$\pi$ accepts~$c$.
In this case, we also say that $\A$ \emph{accepts $c$ with weight
$v_\A(c)$}.
\end{definition}

In \cite{RSJM05} the following problem is considered for the case of
bounded semirings: compute a $\wpds$-automaton $\A$ such that $v_A(c)$
equals the meet-over-all-paths (or equivalently the greatest
fixed point of~(\ref{e:fofc}), which always exists for bounded semirings)
from $c$ to $c_f$, for every configuration~$c$.

We extend this solution to the case of unbounded semirings, using
Theorem~\ref{thm:fmh}. We first apply Algorithm~\ref{alg:computation}
to the equation system~(\ref{e:goft}). If the algorithm yields the
greatest fixed point, then we construct a $\wpds$-automaton
$\A=(P,\stackalph,\mathord{\to},P,\{c_f\})$, with $(p,X,d,q)\in\mathord{\to}$
for all $p,X,q$ such that $d$ is the value of $\sema{pXq}$ in the
greatest fixed point computed by Algorithm~\ref{alg:computation}.
Given a configuration~$c$, it is easy to see
that $v_\A(c)$ yields the same result as in~(\ref{e:vmu}).

\begin{example} \label{ex:wpdsaut}
The automaton arising from Example~\ref{ex:wpdssolv} is depicted
below 
where the witness component is marked by $\undefd$ and transitions
with the value $\infty$ are omitted completely.
\begin{center}
\scalebox{0.3}{\input{automaton.pstex_t}}
\end{center}
\vspace{-10mm}
\qed
\end{example}
The problem of computing successors is also considered in~\cite{RSJM05},
i.e., computing a $\wpds$-automaton $\A$ where $v_\A(c)$ is
the meet-over-all-paths
from an initial configuration $c_0$ to $c$. Using our technique,
this result can also be extended to unbounded
semirings;
Appendix~\ref{app:post}
shows an equation system for this
problem, which can be converted into a $\wpds$-automaton
for $\mathit{post}^*(c_0)$ in analogous
fashion.

%% file: automaton.pstex_t
\begin{picture}(0,0)%
\includegraphics{automaton.pstex}%
\end{picture}%
\setlength{\unitlength}{4144sp}%
\begingroup\makeatletter\ifx\SetFigFontNFSS\undefined%
\gdef\SetFigFontNFSS#1#2#3#4#5{%
  \reset@font\fontsize{#1}{#2pt}%
  \fontfamily{#3}\fontseries{#4}\fontshape{#5}%
  \selectfont}%
\fi\endgroup%
\begin{picture}(3409,1519)(3044,-1067)
\put(3511,-736){\makebox(0,0)[b]{\smash{{\SetFigFontNFSS{25}{30.0}{\rmdefault}{\mddefault}{\updefault}{\color[rgb]{0,0,0}$p$}%
}}}}
\put(5986,-736){\makebox(0,0)[b]{\smash{{\SetFigFontNFSS{25}{30.0}{\rmdefault}{\mddefault}{\updefault}{\color[rgb]{0,0,0}$q$}%
}}}}
\put(4681,-421){\makebox(0,0)[b]{\smash{{\SetFigFontNFSS{17}{20.4}{\rmdefault}{\mddefault}{\updefault}{\color[rgb]{0,0,0}$X,\undefd$}%
}}}}
\put(3151, 29){\makebox(0,0)[rb]{\smash{{\SetFigFontNFSS{17}{20.4}{\rmdefault}{\mddefault}{\updefault}{\color[rgb]{0,0,0}$Y,1$}%
}}}}
\put(6391, 29){\makebox(0,0)[lb]{\smash{{\SetFigFontNFSS{17}{20.4}{\rmdefault}{\mddefault}{\updefault}{\color[rgb]{0,0,0}$Y,-2$}%
}}}}
\end{picture}%

%% file: applications.tex
\section{Applications}
\label{sec:apps}
Here we outline some applications of the
theory developed in this paper.
Unless stated otherwise, we will consider the
semiring $\sint$
as described in Example~\ref{ex:semirings}.
Following Remark~\ref{rem:all-witnesses}
and Corollary~\ref{cor:intring},
we assume that all
nonterminating components can be detected in this semiring and
the corresponding transitions in the $\wpds$-automaton will be
assigned the value $\undefd$. The terminating components
resp. the corresponding transitions in the $\wpds$-automaton
take the computed value.

Note that the previously known approaches to reachability in weighted pushdown
automata are not applicable to any of the below presented cases because
they required the semiring to be bounded (no infinite descending chains).
Boundedness is, however, not satisfied in any of our applications.
Our first two applications are new and we are not aware of any other
algorithms that could achieve the same results. 
Our third application deals with shape-balancedness of context-free
languages, a problem for which an algorithm was recently
described in~\cite{TozawaM07}.


\paragraph{Memory Allocations in Linux Kernel.}
Correct memory allocation and deallocation is crucial for the proper
functionality of an operating system. In Linux the library {\tt linux/gfp.h}
is used for allocation and deallocation of kernel memory pages via
the functions {\tt alloc\_pages} and {\tt \_free\_pages} respectively.
The functions which are argumented with a number $n$ (also called
the \emph{order}) allocate or
deallocate $2^{n}$ memory pages.
Citing~\cite[page 187]{linuxbook}:``You must be careful to free only pages
you allocate. Passing the wrong {\sf struct page} or address, or
the incorrect \emph{order}, can result in corruption.''
This means that a basic safety requirement is: never free more pages than
what are allocated.

As most questions about real programs
are in general undecidable, several techniques
have been suggested to provide more tractable models.
For example so-called boolean programs~\cite{BouajjaniE06} have recently
been used to provide a suitable abstraction via pushdown systems.
Assume a given pushdown system abstraction resulting from the
program code. The transitions in the pushdown system
are labelled with the programming primitives, among others the ones
for allocation and deallocation of memory pages. If a given pushdown
transition allocates  $2^n$ memory pages, we assign it the weight
$2^n$; if it deallocates $2^n$ pages, we assign it the weight $-2^n$;
in all other cases the weight is set to $0$.

Now the pushdown abstraction corrupts the memory iff
a configuration is reachable from the given initial
configuration $pX$ with negative weight.
As shown in Section~\ref{sec:wpds}, we can in polynomial
time (w.r.t.\ to the input pushdown system~$\mathcal{W}$) construct
a $\mathcal{W}$-automaton $\mathcal{A}$ for
$\mathit{post}^*(\{pX\})$. For technical convenience, we first replace
all occurrences of $\undefd$ in $\mathcal{A}$ with $-\infty$.
From all initial control-states of $\mathcal{A}$ we now
run e.g. the Bellman-Ford shortest path algorithm (which can detect negative
cycles and assign the weight to $-\infty$ should there be such)
to check whether there is a path going to some accept state with
an accumulated negative weight. This is doable in polynomial time.
If a negative weight path is found this means that
the corresponding configuration is reachable with a negative weight,
hence there is
a memory corruption (at least in the pushdown abstraction).
Otherwise, the system is safe.
All together our technique gives a polynomial time algorithm for checking
memory corruption with respect to the size of the abstracted pushdown system.
Also depending on whether under- or over-approximation is used
in the abstraction step, our technique can be used for detecting
errors or showing the absence of them, respectively.


\paragraph{Correspondence Assertions.}
\label{app:ca}

In \cite{woo93semantic} Woo and Lam analyze protocols using the so-called
\emph{correspondences} between protocol points. A correspondence property
relates the occurrence of a transition to an earlier occurrence of some
other transition. In sequential programs (modelled as pushdown systems)
assume that assertions of
the form ${\tt begin}~\ell$ and ${\tt end}~\ell$ (where $\ell$ is a label
taken from a finite set of labels) are inserted by the programmer into
the code. The program is \emph{safe} if for each ${\tt end}~\ell$ reached at a
program point there is a unique corresponding ${\tt begin}~\ell$ at an
earlier execution point of the program.
Verifying safety via correspondence assertions can be done using
a similar technique as before. For each label $\ell$ we create
a weighted pushdown system based on the initially given boolean
program abstraction where every instruction ${\tt begin}~\ell$ has
the weight $+1$, every instruction ${\tt end}~\ell$ the weight $-1$,
and all other instructions have the weight $0$. Now the pushdown system
is safe if and only if
every reachable configuration has nonnegative accumulated weight.
This can be verified in polynomial time as outlined above.


\paragraph{Shape-Balancedness of Context-Free Languages.}
In static analysis of programs generating XML strings
and in other XML-related questions, the balancedness
problem has been recently studied
(see e.g.~\cite{BerstelB02,KirkegaardM06,MinamideT06}).
The problem is,
given a context-free language with a paired alphabet of opening and
closing tags, to determine whether every word in the language is
properly balanced (i.e. whether every opening tag has a corresponding
closing tag and vice versa). Tozawa and Minamide recently
suggested~\cite{TozawaM07} a polynomial time algorithm for the problem.
Their involved algorithm consists
of two stages and in the first stage they test for the \emph{shape-balancedness}
property, i.e., if all opening tags as well as closing tags are
treated as of the same sort, is every accepted word balanced? Assume
a given pushdown automaton accepting (by final control-states)
the given context-free language. If we label all opening tags
with weight $+1$ and all closing tags with weight $-1$, the shape-balancedness
question is equivalent
to checking (i) whether every accepted word has the weight
equal to $0$ and (ii) whether all configurations on every path to some
final control-state have nonnegative accumulated weights.
Our generic technique provides polynomial time
algorithms to answer these questions.

To verify property (i), we first consider the
semiring
$\sint = (\integerswithinfty,\min,+,\infty,0)$. We now construct in
polynomial time
for the given initial configuration $pX$
a weighted $\mathit{post}^*(\{pX\})$ $\mathcal{W}$-automaton
$\mathcal{A}$, replace all labels $\undefd$ with $-\infty$,
and for each final control-state $q$ (of the pushdown automaton)
we find in $\mathcal{A}$
a shortest path from $q$ to every accept state of $\mathcal{A}$.
This can be done in polynomial
time using e.g. the Bellman-Ford shortest path algorithm, which can
moreover detect negative cycles and set the respective shortest path
to $-\infty$. If any of the shortest paths are different from $0$,
we terminate because the shape-balancedness property is broken.
If the system passes the first test, we run the same procedure once more
but this time with the semiring $(\integers\cup \{-\infty\},
\max,+,-\infty,0)$ and where $\undefd$ is replaced with
$\infty$, i.e., we are searching for the longest
path in the automaton $\mathcal{A}$. Again if at least one
of those paths has the accumulated weight
different from $0$, we terminate with a negative answer.
If the pushdown system passes both our tests, this means that
any configuration in the set $\mathit{post}^*(\{pX\})$ starting with some
final control-state (of the pushdown automaton)
is reachable only with the accumulated weight $0$
and we can proceed to verify property (ii).

For (ii), we construct
the weighted $\mathit{post}^*(\{pX\})$ $\mathcal{W}$-automaton for
the integer semiring $\sint$. Now we restrict the automaton to
contain only those configurations that can really involve into some
accepting configuration by simply intersecting it (by the usual product
construction)
with the unweighted $\mathcal{W}$-automaton (of polynomial size)
representing $\mathit{pre}^*((q_1 + \cdots + q_n)\Gamma^*)$
where $q_1, \ldots, q_n$ are all final control-states and
$\Gamma$ is the stack alphabet.
Property (ii) now reduces to checking whether the product automaton
accepts some configuration with negative weight, which can be answered
in polynomial time using the technique described in our first
application.

Unfortunately, \cite{TozawaM07} provides no
complexity analysis other than the statement that the algorithm is
polynomial. Our general-purpose algorithm, on the other hand, immediately
provides a precise complexity bound.
Consider a given context-free
grammar of size $n$ over some paired alphabet. It can be (by the standard
textbook construction) translated into a (weighted) pushdown automaton of size
$O(n)$ and moreover with a constant number of states.
As mentioned in Section~\ref{sec:wpds}, this automaton can be normalized
in linear time and we can then build a weighted $\mathit{post}^*(\{pX\})$
$\mathcal{W}$-automaton, of size $O(n^2)$ with $O(n)$ states and in time
$O(n^4)$.
Details can be found in Appendix~\ref{app:post}.
Now running the Bellman-Ford algorithm twice in order to verify property (i)
takes only the time $O(n^3)$. In property (ii) the Bellman-Ford algorithm
is run on a product of the weighted $\mathit{post}^*$ automaton and
an unweighted $\mathit{pre}^*$ automaton, which has only a constant
number states. Hence the size of the product is still $O(n^2)$ and Bellman-Ford
algorithm will run in time $O(n^3)$ as before. This gives the total running time of~$O(n^4)$.

%% file: conclusion.tex
\section{Conclusion}
\label{sec:conclusion}
We presented a unified framework how to deal with interprocedural
dataflow analysis on weighted pushdown automata where the weight
domains might contain infinite descending chains. The problem was
solved by reformulating it
via generalized fixed-point equations which required
polynomials of degree two. To the best of our knowledge this is
the first approach that enables to handle this kind of domains.
On the other hand,
we do not consider completely general idempotent semirings as
we require that the elements in the domain are totally ordered
and that extend preserves inequality. Nevertheless, we showed that
our theory is still applicable. Already the reachability
analysis of weighted pushdown automata over the integer semiring,
one particular instance of our general framework,
was not known before and we provided several examples of its 
potential use in verification.

Regarding the two restrictions we introduced,
we claim that the first
condition of total ordering can be relaxed to orderings of bounded
width, where the maximum number of incomparable elements is bounded
by some a priori given constant $c$. By running the main loop in
Algorithm~\ref{alg:computation} $cn+1$ times, we should be able
to detect nontermination also in this case. The motivation for
introducing bounded width comes from the fact that this will allow us
to combine (via the product construction) one unbounded domain, like
e.g. the integer semiring, with
a fixed number of finite domains in order to observe additional properties
along the computations.
The question whether the second restriction (extend preserves inequality)
can be relaxed as well remains open and is a part of our future work.



%% file: proof-theorem.tex
\section{Proof of Theorem~\ref{thm:dg}}
\label{proof:thm}

The statement about the runtime follows immediately from our definition of~$K(\vf)$.
If Algorithm~\ref{alg:computation} returns a fixed point, it is the greatest fixed point as it is the result of Kleene's iteration.
It remains to show that if the algorithm returns the statement of line~\ref{line:witness} then this statement in fact holds.

For that purpose we introduce the concept of derivation trees that was also used in~\cite{EKL07:stacs,EKL07:dlt}.
It generalizes the well-known notion from language theory to semirings.
In the following we identify a node $x$ of a tree $t$ with the subtree of $t$ rooted at $x$.
In particular, we identify a tree with its root.

\begin{definition}[Derivation Tree]\label{def:devtree}
Let $\vf$ be a vector of $n$ polynomials. 
A {\em derivation tree} $t$ of $\vf$ is an ordered finite tree whose nodes are labelled with both
 a variable $X_i$ ($1 \le i \le n$) and a monomial $m$ of~$\vf_i$.
We write $\lambda_v$, resp.\ $\lambda_m$ for the corresponding labelling-functions.
If $\lambda_m(x) = a_1 X_{i_1} a_2 \ldots X_{i_s} a_{s+1}$ for some $s \ge 0$,
 then $x$ has exactly $s$ children $x_1,\ldots,x_s$, ordered from left to right,
 with $\lambda_v(x_j)=X_{i_j}$ for all $j = 1, \ldots, s$.
\end{definition}

Notice that a node $x$ in a derivation tree is a leaf if and only if $\lambda_m(x) = a$ for some constant $a \in \domain$.
The \emph{height} $h(t)$ of a derivation tree $t$ is the length of a longest path from the root to a leaf.
For the length, we count the number of nodes on the path including both the root and the leaf.
The \emph{yield} $\yi(t)$ of a derivation tree $t$ with  $\lambda_m(t) = a_1X_{i_1}a_2\cdots X_{i_s}a_{s+1}$ is inductively defined as
 $\yi(t) = a_1\yi(t_1)a_2\cdots \yi(t_s)a_{s+1}$.
Figure~\ref{fig:tree} shows a derivation tree for our running example.
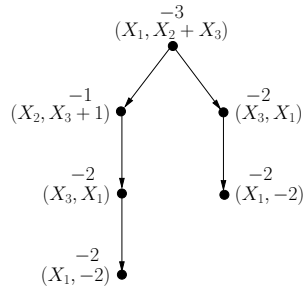
\begin{figure}
\begin{center}
\scalebox{0.3}{\input{tree-running.pstex_t}}
\end{center}
\caption{A derivation tree of height $4$ for $\vf = ( \min\{-2, X_2 + X_3\}, \ X_3 + 1, \ \min\{X_1,X_2\} )$.
 The labels of a node $x$ are denoted by $(\lambda_v(x),\lambda_m(x))$.
 The yields are written on top on the labels.}
\label{fig:tree}
\end{figure}
%

The following proposition
 is easy to prove by induction on the height (see also~\cite{EKL07:dlt}).
\begin{proposition} \label{prop:der-tree-height}
 Let $\vf$ be a vector of $n$ polynomials over a semiring.
 For all $k \in \{1,2,\ldots\}$ and all $1 \le i \le n$ we have
 \[
  \ks{k}_i = \combineBig \left\{ \, \yi(t) \mid t \text{ is a derivation
tree of $\vf$ with } h(t) \le k \text{ and } \lambda_v(t) = X_i \, \right\} \;.
 \]
\end{proposition}
Notice that the set of yields in Proposition~\ref{prop:der-tree-height} is always finite and may be empty.
If it is empty we set $\combineBig \emptyset = \vzero$.
%
Now we prove the following lemma from which the correctness of Algorithm~\ref{alg:computation} follows immediately.

\begin{lemma}
 Let $\vf$ be a vector of $n$ polynomials over a totally ordered idempotent semiring such that extend preserves inequality.
 Let $(\ks{i})_{i\in\naturals}$ denote its Kleene sequence.
 If $\ks{n}_i \ne \ks{n+1}_i$ for some $1 \le i \le n$ then $i$ is a witness component.
\end{lemma}
\begin{proof}
In this proof we write $a \sqsubset b$ to denote that
$a \sqsubseteq b$ and $a \ne b$.
%
We first show the following: 
\begin{equation} \label{eq:claim}
\text{
If $\ks{k}_{i} \sqsubset \ks{k-1}_{i}$ for some $k > n$ then $\ks{k'}_{i} \sqsubset \ks{k}_{i}$ for some $k' > k$.
}
\end{equation}

Let $\ks{k}_{i} \sqsubset \ks{k-1}_{i}$.
By Proposition~\ref{prop:der-tree-height} and using the total
order of the semiring,
 there is a tree~$t$ with $\lambda_v(t) = X_i$ such that $\ks{k}_i = \yi(t)$ and $h(t) = k > n$.
So there is a path in~$t$ from the root to a leaf and some variable $X_j$ 
with two
 nodes $x_1,x_2$ on the path such that $\lambda_v(x_1) = \lambda_v(x_2) = X_j$.
Assume w.l.o.g.\ that $x_1$ is closer to the root than $x_2$.
As $\sqsubseteq$ is a total order, one of the following holds.
\begin{itemize}
 \item
  If $\yi(x_2) \sqsubseteq \yi(x_1)$ then construct a tree~$t'$ from $t$ by replacing the subtree rooted at~$x_1$
   by the subtree rooted at~$x_2$.
  We have $\lambda_v(t') = X_i$ and $h(t') = k'$ for some $k' < k$.
  By monotonicity of $\extend$ (Lemma~\ref{lem:properties} part (ii)) we have $\yi(t') \sqsubseteq \yi(t)$.
  So $\yi(t) = \ks{k}_i \stackrel{\text{Prop.~\ref{prop:kleene}(a)}}{\sqsubseteq} \ks{k'}_i
   \stackrel{\text{Prop.~\ref{prop:der-tree-height}}}{\sqsubseteq} \yi(t') \sqsubseteq \yi(t)$.
  Hence, $\ks{k}_i = \ks{k'}_i$ which, by Prop.~\ref{prop:kleene}(a), implies $\ks{k}_i = \ks{k-1}_i$.
  This contradicts the assumption that $\ks{k-1}_i \ne \ks{k}_i$.
  So this case does not occur.
 \item
  If $\yi(x_1) \sqsubset \yi(x_2)$
   then construct a tree~$t'$ from $t$ by replacing the subtree rooted at~$x_2$ by the subtree rooted at~$x_1$.
  We have $\lambda_v(t') = X_i$ and $h(t') = k'$ for some $k' > k$.
  By monotonicity of $\extend$ (Lemma~\ref{lem:properties} part (ii)) and as extend preserves inequality
   we have $\ks{k'}_i \sqsubseteq \yi(t') \sqsubset \yi(t) = \ks{k}_i$.
  So $\ks{k'}_i \sqsubset \ks{k}_i$.
\end{itemize}
This proves our claim~(\ref{eq:claim}).

It follows from the claim and Proposition~\ref{prop:kleene}(a) that
 if $\ks{k}_{i} \sqsubset \ks{k-1}_{i}$ for some $k > n$ then $\ks{l}_{i} \sqsubset \ks{l-1}_{i}$ for some $l > k$.
Hence, if $\ks{n}_i \ne \ks{n+1}_i$ then $\{ \ks{k}_i \mid k \in \naturals\}$ is infinite.
This completes the proof.
\qed
\end{proof}

%% file: tree-running.pstex_t
\begin{picture}(0,0)%
\includegraphics{tree-running.pstex}%
\end{picture}%
\setlength{\unitlength}{4144sp}%
\begingroup\makeatletter\ifx\SetFigFontNFSS\undefined%
\gdef\SetFigFontNFSS#1#2#3#4#5{%
  \reset@font\fontsize{#1}{#2pt}%
  \fontfamily{#3}\fontseries{#4}\fontshape{#5}%
  \selectfont}%
\fi\endgroup%
\begin{picture}(3135,5452)(7501,-6425)
\put(7876,-4741){\makebox(0,0)[rb]{\smash{{\SetFigFontNFSS{25}{30.0}{\rmdefault}{\mddefault}{\updefault}{\color[rgb]{0,0,0}$(X_3,X_1)$}%
}}}}
\put(10351,-4741){\makebox(0,0)[lb]{\smash{{\SetFigFontNFSS{25}{30.0}{\rmdefault}{\mddefault}{\updefault}{\color[rgb]{0,0,0}$(X_1,-2)$}%
}}}}
\put(10351,-3121){\makebox(0,0)[lb]{\smash{{\SetFigFontNFSS{25}{30.0}{\rmdefault}{\mddefault}{\updefault}{\color[rgb]{0,0,0}$(X_3,X_1)$}%
}}}}
\put(7876,-3121){\makebox(0,0)[rb]{\smash{{\SetFigFontNFSS{25}{30.0}{\rmdefault}{\mddefault}{\updefault}{\color[rgb]{0,0,0}$(X_2,X_3+1)$}%
}}}}
\put(7876,-6361){\makebox(0,0)[rb]{\smash{{\SetFigFontNFSS{25}{30.0}{\rmdefault}{\mddefault}{\updefault}{\color[rgb]{0,0,0}$(X_1,-2)$}%
}}}}
\put(7651,-5956){\makebox(0,0)[rb]{\smash{{\SetFigFontNFSS{25}{30.0}{\rmdefault}{\mddefault}{\updefault}{\color[rgb]{0,0,0}$-2$}%
}}}}
\put(10576,-2761){\makebox(0,0)[lb]{\smash{{\SetFigFontNFSS{25}{30.0}{\rmdefault}{\mddefault}{\updefault}{\color[rgb]{0,0,0}$-2$}%
}}}}
\put(9091,-1096){\makebox(0,0)[b]{\smash{{\SetFigFontNFSS{25}{30.0}{\rmdefault}{\mddefault}{\updefault}{\color[rgb]{0,0,0}$-3$}%
}}}}
\put(7516,-2761){\makebox(0,0)[rb]{\smash{{\SetFigFontNFSS{25}{30.0}{\rmdefault}{\mddefault}{\updefault}{\color[rgb]{0,0,0}$-1$}%
}}}}
\put(7561,-4336){\makebox(0,0)[rb]{\smash{{\SetFigFontNFSS{25}{30.0}{\rmdefault}{\mddefault}{\updefault}{\color[rgb]{0,0,0}$-2$}%
}}}}
\put(9091,-1456){\makebox(0,0)[b]{\smash{{\SetFigFontNFSS{25}{30.0}{\rmdefault}{\mddefault}{\updefault}{\color[rgb]{0,0,0}$(X_1,X_2+X_3)$}%
}}}}
\put(10621,-4336){\makebox(0,0)[lb]{\smash{{\SetFigFontNFSS{25}{30.0}{\rmdefault}{\mddefault}{\updefault}{\color[rgb]{0,0,0}$-2$}%
}}}}
\end{picture}%

%% file: proof-lemma.tex
\section{Proofs of Lemma~\ref{lem:fofc} and Lemma~\ref{lem:goft}}
\label{proof:lem}

Lemma~\ref{lem:fofc} claims that in the equation system~(\ref{e:fofc})
the following holds
for every $k\ge1$ and any configuration~$c$:
$$\ks{k}_{\sema{c}}=\bigoplus\{\,v(\sigma)\mid c\pto{\sigma}c_f, \ |\sigma|<k\,\}$$
This follows directly from 
Proposition~\ref{prop:der-tree-height}, and because every derivation
tree of height~$k$ for~(\ref{e:fofc}) corresponds to a sequence of $k-1$ moves
in the WPDS.\qed

\smallskip

Lemma~\ref{lem:goft} claims that in the equation system~(\ref{e:goft})
the following holds for every $k\ge1$, control states $p,q$, and
stack symbol~$X$:
$$\bigoplus\{\,v(\sigma)\mid c\pto{\sigma}c_f, \ |\sigma|\le 2^{k-1}\,\}
\sqsubseteq
\ks{k}_{\sema{pXq}}
\sqsubseteq
\bigoplus\{\,v(\sigma)\mid c\pto{\sigma}c_f, \ |\sigma|\leq k-1\,\}$$
A derivation tree of height~$k$ for (\ref{e:goft})
corresponds to a path in~$\wpds$ whose length is at least $k-1$ (if all
internal nodes have just one child) and at most $2^{k-1}$ (if all internal
nodes have two children). Because of this, and
because of Proposition~\ref{prop:der-tree-height}, the lemma holds.\qed

%% file: posteq.tex
\section{Computing Successors in Weighted Pushdown Systems}
\label{app:post}

In Section~\ref{sec:wpds}, we considered the following problem:
given a \emph{target} configuration~$c_f$, compute (if possible)
the meet-over-all-paths from $c$ to $c_f$, for any configuration~$c$.
In other words, we considered the \emph{predecessors} of $c_f$.

Alternatively, one could consider the \emph{successors} of some
\emph{source} configuration~$c_s:=p_sX_s$ and attempt to compute the meet
over all paths from $c_s$ to $c$. It is possible to adapt the methods
from Section~\ref{sec:wpds} to this problem (and in fact, this
adaptation is used by our applications).

It is well-known that most results about backward pushdown reachability
carry over to forward pushdown reachability, and vice versa. The easiest
explanation for this is that given a WPDS $\wpds$, one can construct another
WPDS~$\wpds'$ which makes the movements of $\wpds$ `in reverse'. More precisely,
if $\wpds$ has control states $P$, stack alphabet $\Gamma$, and rules~$\Delta$,
then $\wpds'$ has control states
$P':=P\cup\{\,(q,Y)\mid \exists (pX\definesto{d} qYZ)\in\Delta\,\}$,
stack alphabet
$\Gamma\cup\{\#\}$, and the following rules:
\begin{itemize}
\item if $pX\definesto{d} qY\in\Delta$, then $qY\definesto{d} pX\in\Delta'$;
\item if $pX\definesto{d} q\epsilon\in\Delta$, then
	$qY\definesto{d}pXY\in\Delta'$ for every $Y\in\Gamma\cup\{\#\}$;
\item if $pX\definesto{d} qYZ\in\Delta$, then
	$qY\definesto{\extendNeutral}(q,Y)\epsilon$ and $(q,Y)Z\definesto{d}pX$
	in $\Delta'$.
\end{itemize}
It is easy to see that whenever $p\alpha\pto{\sigma}q\beta$ holds in $\wpds$,
then $q\beta\#\pto{\tau}p\alpha\#$ holds for some rule sequence~$\tau$
in~$\wpds'$ such that, if $\sigma=r_1\ldots r_n$ and $\tau=s_1\ldots s_m$,
then $d_{r_1}\extend\cdots\extend d_{r_n}=d_{s_m}\extend\cdots\extend d_{s_1}$.
Thus, it is possible to reduce forward reachability problems to backward
reachability problems, and the reduction is polynomial.

It is also possible to tackle the forward reachability problem directly,
in which case slightly better complexity bounds can be achieved, see,
for instance~\cite{EsparzaHRS00,RSJM05}. Following the ideas
from~\cite{EsparzaHRS00,RSJM05}, we will present a finite equation system
that serves as the `forward analogy' of (\ref{e:goft}), without proof.
Our system has the following sets of variables:
\begin{itemize}
\item $\sema{pX\bullet}$, for $p\in P$ and $X\in\Gamma$,
      representing the weights of the paths from $p_sX_s$ to $pX$;
\item $\sema{pX(rZ)}$, for $p\in P$, $X\in\Gamma$, and $(r,Z)\in P'$,
      representing the weights of the paths from $rZ$ to $pX$;
\item $\sema{p\epsilon\bullet}$, for $p\in P$,
      representing the weights of the paths from $p_sX_s$ to $p\epsilon$;
\item $\sema{p\epsilon(rZ)}$, for $p\in P$ and $(r,Z)\in P'$,
      representing the weights of the paths from $rZ$ to $p\epsilon$;
\item $\sema{(pX)Y\bullet}$, for $(p,X)\in P'$ and $Y\in\Gamma$,
      representing the weights of the paths from $p_sX_s$ to $pXY$,
      ending with a `push' operation;
\item $\sema{(pX)Y(rZ)}$, for $(p,X),(r,Z)\in P'$ and $Y\in\Gamma$,
      representing the weights of the paths from $rZ$ to $pXY$,
      ending with a `push' operation.
\end{itemize}

Moreover, we define
$I(pX)=\extendNeutral$ iff $pX=p_sX_s$ and $\combineNeutral$ otherwise,
and $E(pX,rZ)=\extendNeutral$ iff $pX=rZ$ and $\combineNeutral$
otherwise, for $(p,X)\in P'$. The equation system is as follows:

\begin{eqnarray*}
\sema{pX\bullet} & = & I(pX) \combine \bigoplus_{qY\definesto{d}pX} \big(\sema{qY\bullet}\extend d\big) \combine \bigoplus_{(q,Y)\in P'} \big(\sema{(qY)X\bullet} \extend \sema{p\epsilon(qY)}\big)\\
\sema{pX(rZ)} & = & E(pX,rZ) \combine \bigoplus_{qY\definesto{d}pX} \big(\sema{qY(rZ)}\extend d\big) \combine \bigoplus_{(q,Y)\in P'} \big(\sema{(qY)X(rZ)} \extend \sema{p\epsilon(qY)}\big)\\
\sema{p\epsilon\bullet} & = & \bigoplus_{qY\definesto{d}p\epsilon} \big(\sema{qY\bullet}\extend d\big)\\
\sema{p\epsilon(rZ)} & = & \bigoplus_{qY\definesto{d}p\epsilon} \big(\sema{qY(rZ)}\extend d\big)\\
\sema{(pX)Y\bullet} & = & \bigoplus_{qU\definesto{d}pXY} \big(\sema{qU\bullet}\extend d\big)\\
\sema{(pX)Y(rZ)} & = & \bigoplus_{qU\definesto{d}pXY} \big(\sema{qU(rZ)}\extend d\big)
\end{eqnarray*}

Intuitively, the right-hand sides of the equations list the possible ways
in which the paths corresponding to the left-hand-side variables can be
generated.

In analogy with Section~\ref{sub:wa}, any solution of $\boldsymbol{h}$ can be
converted into a $\mathit{post}^*(\{p_sX_s\})$
$\wpds$-automaton. Our automaton~$\A$ has $\epsilon$-edges,
and its states are $P'$ extended with a final state~$\bullet$. Every variable
$\sema{sXs'}$, where $s,s'\in P'\cup\{\bullet\}$ and
$X\in\Gamma\cup\{\epsilon\}$, and its value
in the solution then correspond to a transition of~$\A$.
The meet-over-all-paths for every configuration~$c$ can be obtained by
identifying the paths on which $c$ is accepted by $\A$ and computing
$v_\A(c)$.

\begin{remark}
\label{rem:post}
According to~\cite{Sch02b,RSJM05}, the size of the equation system
and the number of variables is $\mathcal{O}(|P|\cdot|\Delta|^2)$,
therefore the time for Algorithm~\ref{alg:computation} is
$\mathcal{O}(|P|^2|\Delta|^4)$. The resulting automaton has got
$\mathcal{O}(|P|+|\Delta|)$ states.
\end{remark}